\documentclass[conference]{IEEEtran}
\usepackage{cite}% *** CITATION PACKAGES ***
\usepackage[cmex10]{amsmath}% *** MATH PACKAGES ***
\usepackage{amssymb}
\usepackage{amsthm}
\usepackage{mathrsfs}%SF
\usepackage{bm}
\usepackage[mathscr]{eucal}
\usepackage{amssymb,amsmath,amsthm,amsfonts,latexsym}
\usepackage{amsmath,graphicx,bm,xcolor,url}
\usepackage{graphicx}
\usepackage{latexsym}%SF
\usepackage{CJK}%SF
\usepackage{indentfirst}%SF
\usepackage{geometry}%SF
\usepackage{psfrag}%SF
\usepackage{setspace}%SF
\usepackage{algorithmic}% *** SPECIALIZED LIST PACKAGES ***
\usepackage{algorithmic, cite}
\usepackage{algorithm}
\usepackage{array}% *** ALIGNMENT PACKAGES ***
\usepackage{mdwmath}
\usepackage{mdwtab}
\usepackage{eqparbox}
\usepackage{url}% *** PDF, URL AND HYPERLINK PACKAGES ***
\usepackage{epstopdf}
\usepackage{epsfig,epsf,psfrag}
\usepackage{fixltx2e}%ordering of single and double column floats
\usepackage{verbatim}
\usepackage{textcomp}
\usepackage{psfrag} %%25 June

\usepackage{subfigure} %07-30
\usepackage{caption}
\theoremstyle{plain}
\newtheorem{mythe}{Theorem}
\theoremstyle{remark}

\theoremstyle{plain}

\theoremstyle{remark}

\theoremstyle{plain}

\theoremstyle{remark}

\theoremstyle{remark}

\theoremstyle{remark}

\theoremstyle{remark}

\theoremstyle{remark}

\theoremstyle{remark}

\catcode`~=11 \def\UrlSpecials{\do\~{\kern -.15em\lower .7ex\hbox{~}\kern .04em}} \catcode`~=13

\allowdisplaybreaks[4]

\newcommand{\calC}{\mathcal{C}}

\newcommand{\calN}{\mathcal{N}}

\newcommand{\bp}{\mathbf{p}}
\newcommand{\bP}{\mathbf{P}}

\newcommand{\bbE}{\mathbb{E}}

\DeclareMathAlphabet{\mathbsf}{OT1}{cmss}{bx}{n}
\DeclareMathAlphabet{\mathssf}{OT1}{cmss}{m}{sl}% slanted sans serif

\DeclareSymbolFont{bsfletters}{OT1}{cmss}{bx}{n}
\DeclareSymbolFont{ssfletters}{OT1}{cmss}{m}{n}
\DeclareMathSymbol{\bsfGamma}{0}{bsfletters}{'000}
\DeclareMathSymbol{\ssfGamma}{0}{ssfletters}{'000}
\DeclareMathSymbol{\bsfDelta}{0}{bsfletters}{'001}
\DeclareMathSymbol{\ssfDelta}{0}{ssfletters}{'001}
\DeclareMathSymbol{\bsfTheta}{0}{bsfletters}{'002}
\DeclareMathSymbol{\ssfTheta}{0}{ssfletters}{'002}
\DeclareMathSymbol{\bsfLambda}{0}{bsfletters}{'003}
\DeclareMathSymbol{\ssfLambda}{0}{ssfletters}{'003}
\DeclareMathSymbol{\bsfXi}{0}{bsfletters}{'004}
\DeclareMathSymbol{\ssfXi}{0}{ssfletters}{'004}
\DeclareMathSymbol{\bsfPi}{0}{bsfletters}{'005}
\DeclareMathSymbol{\ssfPi}{0}{ssfletters}{'005}
\DeclareMathSymbol{\bsfSigma}{0}{bsfletters}{'006}
\DeclareMathSymbol{\ssfSigma}{0}{ssfletters}{'006}
\DeclareMathSymbol{\bsfUpsilon}{0}{bsfletters}{'007}
\DeclareMathSymbol{\ssfUpsilon}{0}{ssfletters}{'007}
\DeclareMathSymbol{\bsfPhi}{0}{bsfletters}{'010}
\DeclareMathSymbol{\ssfPhi}{0}{ssfletters}{'010}
\DeclareMathSymbol{\bsfPsi}{0}{bsfletters}{'011}
\DeclareMathSymbol{\ssfPsi}{0}{ssfletters}{'011}
\DeclareMathSymbol{\bsfOmega}{0}{bsfletters}{'012}
\DeclareMathSymbol{\ssfOmega}{0}{ssfletters}{'012}

\newcommand{\tilR}{\widetilde{R}}

\newcommand{\tilW}{\widetilde{W}}

\newcommand{\hatX}{\widehat{X}}

\newcommand{\barP}{\bar{P}}

\newcommand{\balpha}{\bm{\alpha}}

\newcommand{\btau}{\bm{\tau}}

\def\norm#1{\left\| #1 \right\|}
\def\norm2#1{\left\| #1 \right\|_2}
\def\norm22#1{\left\| #1 \right\|_2^2}

\newcommand{\eqa}{\stackrel{(a)}{=}}

\newcommand{\leb}{\stackrel{(b)}{\le}}
\newcommand{\lec}{\stackrel{(c)}{\le}}

\newcommand{\qednew}{\nobreak \ifvmode \relax \else
      \ifdim\lastskip<1.5em \hskip-\lastskip
      \hskip1.5em plus0em minus0.5em \fi \nobreak
      \vrule height0.75em width0.5em depth0.25em\fi}

\title{Optimal Resource Allocation in Full-Duplex Ambient Backscatter Communication Networks for Green IoT}

\author{Gang~Yang, \emph{Member, IEEE}, Dongdong Yuan, and Ying-Chang~Liang, \emph{Fellow, IEEE}\\
University of Electronic Science and Technology of China, Chengdu, P. R. China \\
Email: yanggang@uestc.edu.cn, 201721260420@std.uestc.edu.cn, liangyc@ieee.org}

\begin{document}
 \maketitle

\begin{abstract}
Ambient backscatter communication (AmBC) enables wireless-powered backscatter devices (BDs) to transmit information over ambient radio-frequency (RF) carriers without using an RF transmitter, and thus has emerged as a promising technology for green Internet-of-Things. This paper considers an AmBC network in which a full-duplex access point (FAP) simultaneously transmits downlink orthogonal frequency division multiplexing (OFDM) signals to its legacy user (LU) and receives uplink signals backscattered from multiple BDs in a time-division-multiple-access manner. To enhance the system performance from multiple design dimensions and ensure fairness, we maximize the minimum throughput among all BDs by jointly optimizing the BDs' backscatter time portions, the BDs' power reflection coefficients, and the FAP's subcarrier power allocation, subject to the LU's throughput constraint, the BDs' harvested-energy constraints, and other practical constraints. As such, we propose an efficient iterative algorithm for solving the formulated non-convex problem by leveraging the block coordinated decent and successive convex optimization techniques. We further show the convergence of the proposed algorithm, and analyze its complexity. Finally, extensive simulation results show that the proposed joint design achieves significant throughput gains as compared to the benchmark scheme with equal resource allocation.
\end{abstract}
%The formulated problem is a non-convex optimization problem with mutually coupled constraints, and thus challenging to solve.
%complex and power-hungry
%energy-efficient and cost-efficient

\vspace{-0.2cm}
\section{Introduction}
% which covers various use cases like smart home, smart healthcare, and smart city
Internet of Things (IoT) is a key application scenario for the fifth generation (5G) and future mobile communication systems, and various IoT devices typically have strict limitations on energy, cost, and complexity. Recently, ambient backscatter communication (AmBC) which enables backscatter devices (BDs) to modulate their information symbols over the ambient RF carriers (e.g., WiFi, TV, or cellular signals) without using a complex and power-hungry RF transmitter~\cite{ABCSigcom13}, has emerged as a promising technology for energy-efficient and cost-efficient IoT communications.
%In contrast to traditional backscatter communication systems~\cite{BoyerSumit14}, AmBC does not require any dedicated infrastructure (e.g., reader, carrier emitter) to generating RF sinusoidal carriers.
%Dobkinbook2007,  like radio-frequency identification (RFID) systems

%\footnote{Hereinafter, the term ``legacy'' refers to existing wireless communication systems like WiFi systems.}
The existing AmBC systems can be divided into three categories, the standard AmBC system with separated backscatter receiver and ambient transmitter (its legacy receiver) \cite{WiFiBackscatter14,YangLiangZhangPeiTCOM17,WangGaoAmBCTCOM16,FSBackscatterSigcomm16, WangSmithFMBackscatter17}, the cooperative AmBC system with co-located backscatter receiver and  legacy receiver \cite{YangLiangZhangIoTJ18,LongYangGC17}, and the full-duplex AmBC system with co-located backscatter receiver and ambient transmitter \cite{BackFiSigcom15}. Existing works on AmBC focus on the transceiver design and hardware prototype. To address the problem of strong direct-link interference from ambient transmitter in standard AmBC systems, some studies shift the backscattered signal to a clean band that does not overlap with the direct-link signal~\cite{FSBackscatterSigcomm16, WangSmithFMBackscatter17}. In \cite{YangLiangZhangPeiTCOM17}, the BD waveform and backscatter receiver detector are jointly designed to cancel out the direct-link interference from the ambient orthogonal frequency division multiplexing (OFDM) signals. A full-duplex AmBC system is proposed in \cite{BackFiSigcom15}, in which the WiFi access point decodes the received backscattered signal while simultaneously transmitting WiFi packages to its legacy client.
%, which thus need additional spectrum  to provide Internet connectivity for RF-powered BD via existing WiFi infrastructurep
%, by exploiting the repeating structure of the ambient orthogonal frequency division multiplexing (OFDM) signals due to the use of cyclic prefix
%The strong direct-interference from the RF source is a key challenge for transceiver design for standard AmBC.
%This design relies on the self-interference-cancellation for full-duplex operation at the AP.
%InterscatterSigcom16,

%by optimizing the legacy transmitter's subcarrier power and the BD's reflection coefficient
Recently, there are a few literature on performance analysis and optimization for AmBC systems. In \cite{KangLiangICC17}, the ergodic capacity for backscatter communication is maximized for a standard AmBC system. In \cite{LongYangGC17}, the transmit beamforming is optimized for a cooperative AmBC system with multiple antennas at the ambient transmitter. In~\cite{HoangNiyatoAmBCTCOM17}, for an ambient-backscatter assisted cognitive radio network, the secondary transmitter's rate is maximized by optimizing the time resource.%wireless-powered
%allocation among ambient backscatter, energy harvesting, and active transmission.
%However, to our best knowledge, these aforementioned works still but lack fundamental analysis and performance optimization for full-duplex AmBC systems.

In this paper, we consider a full-duplex AmBC network (ABCN) over ambient OFDM carriers as shown in Fig.~\ref{fig:Fig1}, consisting of a full-duplex access point (FAP), a legacy user (LU), and multiple BDs. We optimize the throughput performance for such an ABCN, which has not been studied in the literature to our best knowledge. The main contributions of this paper are summarized as follows:
%with two antennas for simultaneous signal transmission and reception, respectivelyp
%One typical application examples is described as follows: a WiFi AP simultaneously transmits downlink information via OFDM modulation to its client (e.g., smartphone, laptop) and receives uplink information from multiple domestic IoT devices (e.g., tags, sensor) for smart-home applications.
\begin{itemize}
  \item First, by employing an FAP, we propose a new model to enable simultaneous downlink information transmission (energy transfer) to the LU (multiple BDs) and uplink information transmission from multiple BDs. We characterize the corresponding throughput and energy transfer performances of the BDs, as well as the throughput performance of the LU.
      % in a time-division-multiple-access (TDMA) mannerp
      %The BD symbol period is designed to be the same as the OFDM symbol period. By performing maximum-ratio-combining (MRC) over subcarriers at the FAP, the throughput expression of each BD is derived. The total throughput of the LU is obtained by treating the received signal backscattered from BDs as noise. Also, the harvested energy of each BD is obtained.
\item Second, to ensure fairness, we formulate a problem to maximize the minimum throughput among all BDs by jointly optimizing three blocks of variables including the BDs' backscatter time portions, the BDs' power reflection coefficients, and the FAP's subcarrier power allocation. Through joint optimization, the system performance can benefit from multiple design dimensions.
    %Such a joint optimization problem is practically appealing.
    %
    %However, the formulated problem is non-convex and thus non-trivial to solve in general, since all variables are mutually coupled closely in the constraints.
    %, subject to the LU's throughput requirement and the BDs' energy-transfer constraints
    %On one hand, by properly designing the power reflection coefficients of near BDs, more time can be allocated for the backscatter transmission of far BDs to further enhance their throughput. On the other hand, by properly allocating transmission power among subcarriers, the optimal throughput tradeoff can be achieved among the BDs and the LU.
  \item Third, to solve the formulated non-convex problem, we propose an iterative algorithm by leveraging the block coordinated decent (BCD) and successive convex optimization (SCO) techniques, in each iteration of which the three blocks of variables are alternately optimized. Also, we show the convergence of the proposed algorithm and analyze its complexity.
      %pHowever, for the non-convex subcarrier power optimization problem, we apply to solve it approximately.
      %The entire optimization variables are partitioned into three blocks for the BDs' backscatter time portions, the BDs' power reflection coefficients, and the FAP's subcarrier power, respectively.
      %The complexity of the proposed algorithm is also analyzed.Specifically, t with given backscatter time portions and power reflection coefficients
      %, i.e., one block is optimized at each time while keeping the other two blocks fixed
  \item Finally, numerical results show that significant throughput gains are achieved by our proposed joint design, as compared to the benchmark scheme with equal resource allocation. Also, the BD-LU throughput tradeoff and the BDs' throughput-energy tradeoff are revealed.
      %shown that the BDs' max-min throughput increases as the LU's throughput requirement decreases or the BDs' energy requirement decreases,
      %In addition, compared to the benchmark case, these tradeoffs are shown to be significantly improved by the use of joint optimization.
\end{itemize}

The rest of this paper is organized as follows. Section~\ref{systemmodel} presents the system model for a full-duplex ABCN over ambient OFDM carriers. Section~\ref{formulation} formulates the minimum-throughput maximization problem. Section~\ref{solution} proposes an efficient iterative algorithm by applying the BCD and SCO techniques. Section~\ref{simulation} presents the numerical results. Finally, Section~\ref{conslusion} concludes this paper.
% to verify the performance of the proposed joint design

\vspace{-0.2cm}
\section{System Model}\label{systemmodel}
%full-duplex access point (a backscatter transmitter (i.e.,
In this section, we present the system model for a full-duplex ABCN. As illustrated in Fig.~\ref{fig:Fig1}, we consider two co-existing communication systems: the legacy communication system which consists of an FAP with two antennas for simultaneous information transmission and reception, respectively, and its dedicated LU\footnote{We consider the case of a single LU, since the FAP typically transmits to an LU in a short period for practical OFDM systems like WiFi system. The analyses and results can be extended to the multiple-LU case.}, and the AmBC system which consists of the FAP and $M$ BDs. The FAP transmits OFDM signals to the LU. We are interested in the AmBC system in which each BD transmits its modulated signal back to the FAP over its received ambient OFDM carrier. Each BD contains a backscatter antenna, a switched load impedance, a micro-controller, a signal processor, an energy harvester, and other modules (e.g., battery, memory, sensing). The BD modulates its received ambient OFDM carrier by intentionally switching the load impedance to vary the amplitude and/or phase of its backscattered signal, and the backscattered signal is received and finally decoded by the FAP. The energy harvester collects energy from ambient OFDM signals and uses it to replenish the battery which provides power for all modules of the BD.
%To transmit information bits stored in the memory, t
%Also, the BD antenna can be switched to the signal processor when the BD needs to receive controlling signal.
\begin{figure}[!t]
\vspace{-0.2cm}
\centering
\includegraphics[width=.98\columnwidth] {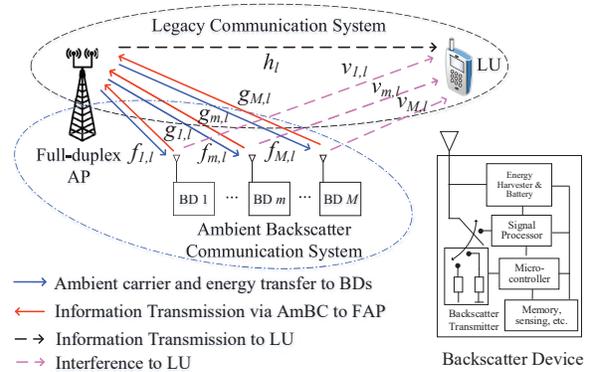}
\caption{System description for a full-duplex ABCN.}%ambient backscatter communication networks
\label{fig:Fig1}
\vspace{-0.2cm}
\end{figure}
%\footnote{Although there can be multiple LUs in the legacy system, we consider the scenario of a single LU. This is because that the FAP typically transmit to a single LU in a short period which is widely adopted in practical OFDM systems like WiFi systems. Notice that the analysis and results can be extended to the case of multiple LUs}

The block fading channel model is assumed. As shown in Fig.~\ref{fig:Fig1}, let $f_{m, l}$ be the $L_{\sf f}$-path forward channel response from the FAP to the $m$-th BD, for $m=1,\ \ldots, \ M$, $g_{m, l}$ be the $L_{\sf g}$-path backward channel response from the $m$-th BD to the FAP, $h_{l}$ be the $L_{\sf h}$-path direct-link channel response from the FAP to the LU, and $v_{m, l}$ be the $L_{\sf v}$-path interference channel response from the $m$-BD to the LU. Let $N$ be the number of subcarriers of the transmitted OFDM signals. For each channel, define the frequency response at the $k$-th subcarrier as $F_{m,k} = \sum_{l=0}^{L_{\sf f}-1} f_{m,l}e^{\frac{-j 2\pi k l}{N}}$, $G_{m,k} = \sum_{l=0}^{L_{\sf g}-1} g_{m,l}e^{\frac{-j 2\pi k l}{N}}$, $V_{m,k} = \sum_{l=0}^{L_{\sf v}-1} v_{m,l}e^{\frac{- j 2 \pi k l }{N}}$, $H_{k} = \sum_{l=0}^{L_{\sf h}-1} h_{l}e^{\frac{- j 2 \pi k l }{N}}$, for $k = 0,\ldots ,N-1$.
%Similarly, for the backward channel from the $m$-th BD to the FAP, define its subcarrier response as; for the interference channel from the $m$-BD to the LU, define the subcarrier response as; and for the direct-link channel from the FAP to the LU, define the subcarrier response as.

\begin{figure}[!t]
\vspace{-0.1cm}
\centering
\includegraphics[width=.98\columnwidth] {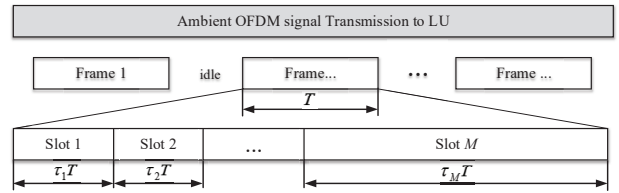}
\caption{Frame structure for TDMA-based ABCN.}
\label{fig:Fig2}
\vspace{-0.7cm}
\end{figure}

%, where the channel coefficient remains the same within each block but may change among blocks. We assume that the channel block length is much longer than the OFDM symbol period ambient backscatter transmission in a

%for ambient backscatter transmission over one block from BDs to the FAP
We consider frame-based transmission, and the frame structure is shown in Fig.~\ref{fig:Fig2}. In each frame of time duration $T$ (seconds) consisting of $M$ slots, the FAP simultaneously transmits downlink OFDM signals to the LU, and receives uplink signals backscattered from all BDs in a time-division-multiple-access (TDMA) manner. The $m$-th slot of time duration $\tau_m T$ (with time portion $\tau_m$ ($0 \leq \tau_m \leq 1$)) is assigned to the $m$-th BD. Denote the backscatter time portion vector $\btau=[\tau_1 \ \tau_2 \ \ldots \ \tau_M]^T$. In the $m$-th slot, BD $m$ reflects back a portion of its incident signal for transmitting information to the FAP and harvests energy from the remaining incident signal, and all other BDs only harvest energy from their received OFDM signals.
% using its transmitting antennap
%, and the time portions are subject to the constraints $\sum_{m=1}^M \tau_m \leq 1$ and $0 \leq \tau_m \leq 1,\ \forall m$
%In the $m$-th slot, a portion of the incident signal power can be harvested by the $m$-th BD.
%total time For convenience, , for $m=1,\ \ldots, M$ $\alpha_m$ ($0 \leq \alpha_k \leq 1$)

%We assume that all BDs transmit data to the FAP in a time-division-multiplexing-access (TDMA) manner,
%For convenience, the frame period $T_{\sf f}$ is normalized to one

Let $S_{m,k}(n) \in \calC$ be the FAP's information symbol at the $k$-th subcarrier, $\forall k$, in the $n$-th OFDM symbol period of the $m$-th slot. After inverse discrete Fourier transform at the FAP, a cyclic-prefix (CP) of length $N_{\sf cp}$ is added at the beginning of each OFDM symbol. The transmitted time-domain signal in each OFDM symbol period is %(IDFT)pgiven by
\begin{align}
  s_{m, t} (n) \!=\! \frac{1}{N}\sum_{k=0}^{N-1} \sqrt{P_{m,k}} S_{m, k} (n) e^{j 2\pi \frac{kt}{N}},
\end{align}
for the time index $t =0, 1, \ldots, N -1$, where $P_{m,k}$ is the allocated power at the $k$-th subcarrier in the $m$-th slot. The subcarrier power values are subject to the average power constraint $\sum \nolimits_{m=1}^M \tau_m \sum \nolimits_{k=0}^{N-1} P_{m,k} \leq \barP$, where  $\barP$ is the total transmission power in all slots. Denote the subcarrier power allocation matrix $\bP=[\bp_1 \ \bp_2 \ \ldots \ \bp_M]$, where $\bp_m$ is the subcarrier power allocation vector in the $m$-th slot.
%\begin{align} For convenience,
%\sum \nolimits_{m=1}^M \tau_m \sum \nolimits_{k=0}^{N-1} P_{m,k} \leq \barP,
%\end{align}

In the $m$-th slot, the incident signal at BD $m$ is $s_{m, t} (n) \otimes f_{m,l}$, where $\otimes$ means the convolution operation. Let $\alpha_m$ ($0 \leq \alpha_m \leq 1$) be the $m$-th BD's power reflection coefficient, and denote the vector $\balpha=[\alpha_1 \ \alpha_2 \ldots \alpha_M]^T$. Let $\eta$ ($0 \leq \eta \leq 1$) be the energy-harvesting efficiency for all BDs. From the aforementioned system model and~\cite{ZhouZhangWirelessCom14}, the total harvested energy by BD $m$ in all slots is thus
\vspace{-0.2cm}
\begin{align}
  &E_m (\btau, \alpha_m, \bP) =  \nonumber \\
  &\eta\sum_{k=0}^{N-1}  |F_{m,k}|^{2} \big[ \tau_{m}  P_{m,k} (1-\alpha_m) + \sum_{r=1, \neq m}^M  \tau_{r} P_{r, k} \big].
%&\eqa \sum_{k=0}^{N-1} \tau_{m} \eta  |F_{m,k}|^{2} P_{m,k} (1-\alpha_m) +  \sum_{r=1, \neq m}^M  \tau_{r} \eta \sum_{k=0}^{N-1}  |F_{m,k}|^{2} P_{r,k} \nonumber \\
\end{align}%\quad \forall m,
% in the $m$-th slot can where the first term in (a) is the harvested energy in the $m$-th slot, and the second term is the harvested energy in all other slots.
\vspace{-0.3cm}

Let $X_m(n)$ be the $m$-th BD's information symbol, whose duration is designed to be the same as the OFDM symbol period. We assume each BD can align the transmission of its own symbol $X_m(n)$ with its received OFDM symbol\footnote{BD can practically estimate the arrival time of OFDM signal by some methods like the scheme that utilizes the repeating structure of CP~\cite{YangLiangZhangPeiTCOM17}.}. In the $m$-th slot, the backscattered signal from the $m$-th BD is thus $\sqrt{P_{m,k}} \sqrt{\alpha_m} s_{m, t} (n) \otimes f_{m,l} X_m(n)$.

The transmitted downlink signal $s_{m,t}$ is known by the FAP's receiving chain. Thus, this signal can be reconstructed and subtracted from the received signals. Therefore, the self-interference can be cancelled by using existing digital or analog cancellation techniques. For this reason, we assume perfect self-interference cancellation (SIC) at the FAP in this paper. After performing SIC, the received signal backscattered from the $m$-th BD is given by
\vspace{-0.1cm}
\begin{align}
  &y_{m, t} (n)=\\
&\quad \sqrt{P_{m,k}} \sqrt{\alpha_m} s_{m, t} (n) \otimes f_{m,l} \otimes g_{m,l} X_m(n) + w_{m, t} (n), \nonumber
\end{align}
\vspace{-0.1cm}where the additive white Gaussian noise (AWGN) is assumed, i.e., $w_{m,t}(n) \sim \calC \calN(0, \sigma^2)$.

After CP removal and discrete Fourier transform operation at the FAP, the received frequency-domain signal is %can be written as follows
\vspace{-0.4cm}
\begin{align}\label{eq:FD_RX-AP}
  &Y_{m, k} (n) = \\
  &\quad \sqrt{P_{m,k}} \sqrt{\alpha_m} F_{m, k} G_{m,k} S_{m, k} (n) X_m (n) + W_{m, k} (n), \nonumber
\end{align}
where the frequency-domain noise $W_{m,k} (n) \sim \calC \calN(0, \sigma^2)$.

The FAP performs maximum-ratio-combining (MRC) to recover the BD symbol $X_m(n)$ as follows
\begin{align}
  \hatX_m (n) =\frac{1}{N} \sum \limits_{k=0}^{N-1} \frac{Y_{m, k} (n)}{\sqrt{P_{m,k}} \sqrt{\alpha_m} F_{m,k} G_{m,k}S_{m,k}},
\end{align}
\vspace{-0.1cm}
and the resulting decoding signal-to-noise-ratio (SNR) is %given as follows
\begin{align}
\gamma_m (\alpha_m, \bP) = \frac{\alpha_m}{\sigma ^{2}} \sum \limits_{k=0}^{N-1} |F_{m,k}|^{2}|G_{m,k}|^{2}P_{m,k}.
\end{align}
\vspace{-0.1cm}

%\footnote{Notice that the throughput is referred as the net rate, since the total time of each frame is normalized to one.}
Hence, the $m$-th BD's throughput\footnote{This paper adopts normalized throughput with unit of bps/Hz.} normalized to $T$ is %given by the frame duration For convenience, t
\begin{align}\label{eq:RateBDm}
  &R_m (\tau_m, \alpha_m, \bp_m)= \nonumber \\
  & \frac{\tau_{m}}{N} \log \left(1+ \frac{\alpha_m}{\sigma^{2}} \sum_{k=0}^{N-1} |F_{m,k}|^{2}|G_{m,k}|^{2}P_{m,k} \right).
\end{align}

Similar to \eqref{eq:FD_RX-AP}, the received frequency-domain signal at the LU can be written as follows
\begin{align}\label{eq:FD_RX-LU}
 & Z_{m, k} (n) = \sqrt{P_{m,k}} H_{k} S_{m, k} (n)  +... \\
 &\sqrt{P_{m,k}} \sqrt{\alpha_m} F_{m, k} V_{m,k} S_{m, k} (n) X_m (n) + \tilW_{m, k} (n), \forall k, m \nonumber
\end{align}
where the frequency-domain noise $\tilW_{m,k} (n) \sim \calC \calN(0, \sigma^2)$.

Similar to~\eqref{eq:RateBDm}, treating backscatter-link signal as interference, the total throughput of the LU is given by
\begin{align}\label{eq:RateLU}
  &\tilR(\btau, \balpha, \bP) = \\
  &\frac{1}{N}\sum_{m=1}^M \tau_m \sum_{k=0}^{N-1}  \log \left(1+\frac{|H_{k}|^2 P_{m,k} }{\alpha_m  |F_{m,k}V_{m,k}|^2 P_{m,k} + \sigma^2}\right). \nonumber
\end{align}

%\newpage
%\section{Max-Min Throughput Optimization}%MULTI-USER RATE REGION
%In this section,we study the joint backscatter time portions, power reflection coefficients, and subcarrier power allocation to maximize the minimum throughput among all BDs.
%%Then,we propose algorithms solve the problem.p
\section{Problem Formulation}\label{formulation}
In this section, we formulate the optimization problem. The objective is to maximize the minimum throughput $Q(\btau, \balpha, \bP) \triangleq \underset{m=1,\ldots,M} {\min} R_m(\tau_m, \alpha_m, \bp_m)$ among all BDs, by jointly optimizing the BDs' backscatter time portions  (i.e., $\btau$), the BD's power reflection coefficients (i.e., $\balpha$), and the FAP's subcarrier power allocation (i.e., $\bP$). We consider the following constraints: the total throughput of the LU needs to be larger than a given minimum throughput $D$, i.e., $\tilR(\btau, \balpha, \bP) \geq D$; each BD has a minimum energy requirement $E_{\min, m}$, i.e., $E_m (\btau, \alpha_m, \bP) \geq E_{\min, m},\forall m$; the total power consumed by the FAP needs to be less than a given maximum power $\barP$, i.e., $  \sum_{m=1}^M  \sum_{k=0}^{N-1} \tau_{m}P_{m,k}\leq \barP$; the sum of backscatter time portions for all BDs should be no larger than 1, i.e., $\sum_{m=1}^M \tau_{m}\leq 1$, with non-negative time portion $\tau_m$ for each BD $m$; the peak power value for each subcarrier is $P_{\sf peak}$, i.e., $0 \leq P_{m,k} \leq P_{\sf peak}, \ \forall m, k$; the power reflection coefficients are positive numbers and no larger than 1, i.e., $0 \leq \alpha_m \leq 1, \ \forall m$. The optimization problem is thus formulated as follows
\begin{subequations}
\label{eq:P1}
\begin{align}
&\underset{Q, \btau, \balpha, \bP}{\max}  \quad Q \\
& \text{s.t.} \quad \frac{\tau_{m}}{N} \!\log\! \!\left(\!1 \!+\ \! \frac{\alpha_m}{\sigma^{2}} \sum_{k=0}^{N-1} |F_{m,k}|^{2}|G_{m,k}|^{2}P_{m,k} \!\right)\!  \!\geq\! Q,  \forall m \label{eq:C1P1}\\
& \frac{1}{N} \!\sum_{m=1}^M\! \tau_m \sum_{k=0}^{N-1}  \!\log\! \!\left(\!1 \!+\! \frac{|H_{k}|^2 P_{m,k} }{\alpha_m  |F_{m,k}V_{m,k}|^2 P_{m,k} \!+\! \sigma^2}\!\right)\!  \!\geq\! D \label{eq:C2P1} \\
&\eta \!\sum_{k=0}^{N-1}\!  |F_{m,k}|^{2} \!\big[\! \tau_{m}  P_{m,k} ( 1\!-\! \alpha_m ) \!+\! \!\sum_{r=1, \neq m}^M\!  \tau_r P_{r,k} \!\big]\! \nonumber \\
& \quad \quad \quad \quad \quad \!\geq\! E_{\min, m},\forall m \label{eq:C3P1} \\
& \sum_{m=1}^M  \sum_{k=0}^{N-1} \tau_{m}P_{m,k}\leq \barP \label{eq:C4P1}\\
& \sum_{m=1}^M \tau_{m}\leq 1 \label{eq:C5P1} \\
& 0 \leq P_{m,k} \leq  P_{\sf peak}, \quad \forall m, \ k \label{eq:C6P1} \\
&  \tau_{m}\geq  0, \quad \forall m \label{eq:C7P1}\\
&  0 \leq \alpha_m \leq  1, \quad \forall m. \label{eq:C8P1}
\end{align}
\end{subequations}
Notice that problem \eqref{eq:P1} is non-convex and challenging to solve in general, since the variables are all coupled and the constraint function in \eqref{eq:C2P1} is non-convex over $P_{m,k}$'s.
%with respect to
%Problem (P1) is challenging to solve due to the following two reasons. First, the backscatter-time portion variables $\{\tau_{m}\}$'s, the power reflection coefficients variables $\{\alpha_m\}$'s and the subcarrier power variables $\{P_{m,k}\}$'s are all coupled in the constraints \eqref{eq:C1P1}, \eqref{eq:C2P1}, \eqref{eq:C3P1}, and \eqref{eq:C4P1}. Second, the logarithm function in the constraint \eqref{eq:C2P1} is a non-convex function of the subcarrier power variables $\{P_{m,k}\}$'s. Therefore, (P1) is a non-convex problem, which is difficult to solve in general.

%Problem (P1) is challenging to solve due to the following two reasons. First, the backscatter-time portions variables $\{\tau_{m}\}$'s, the power reflection coefficients variables $\{\alpha_m\}$'s and the subcarrier power variables $\{P_{m,k}\}$'s are all coupled in the constraints \eqref{eq:C1P1}, \eqref{eq:C2P1}, \eqref{eq:C3P1}, and \eqref{eq:C4P1}. Second, the logarithm function in the constraint \eqref{eq:C2P1} is a non-convex function of the subcarrier power variables $\{P_{m,k}\}$'s. Therefore, (P1) is a non-convex problem, which is difficult to solve in general.
%We will obtain the solutions in the next subsection. %backscatter time portions allocation in general

\vspace{-0.2cm}
\section{Proposed Algorithm}\label{solution}% for Problem (P1)
%In general, there is no standard method for solving the non-convex optimization problem (P1) efficiently.
In this section, we propose an efficient iterative algorithm for the problem \eqref{eq:P1} by applying the block coordinate descent (BCD)~\cite{HongLuoBCDSPM17} and successive convex optimization (SCO) techniques~\cite{Beck2010}. Then, we show the convergence of the proposed algorithm and analyze its complexity.%~\cite{TsengBCD01}
%In each iteration, we optimize different blocks of variables alteratively. Specifically, for any given power reflection coefficients $\balpha$ and subcarrier power $\bP$, we optimize the backscatter time portions $\btau$ by solving a linear programming (LP); for any given backscatter time portions $\btau$ and subcarrier power $\bP$, the power reflection coefficients $\balpha$ is optimized by solving a convex problem; for any given backscatter time portions $\btau$ and power reflection coefficients $\balpha$, the subcarrier power $\bP$ is optimized by utilizing the SCA technique and solving a convex problem. After presenting the overall algorithm
\vspace{-0.2cm}
\subsection{Backscatter Time Allocation Optimization}
%In this subsection, we propose algorithms to solve Problem (P3). Since the power allocation variables $\{P_{m,k}\}$'s and the time allocation variables $\{\tau_{m}\}$'s are coupled in the constraints, we adopt the block-coordinated-decent (BCD) based algorithm to optimize $\{P_{m,k}\}$'s and $\{\tau_{m}\}$'s alteratively in each iteration.
In iteration $j$, for given power reflection coefficients $\balpha^{\{j\}}$ and subcarrier power allocation $\bP^{\{j\}}$,  the backscatter time portions $\btau$ can be optimized by solving the problem %following
%\begin{subequations}
%\label{eq:P1BTA}
%\begin{align}
%{\text{(P1-BTA)}}:\underset{Q, \btau}{\max}  &\quad Q \label{eq:ObjP1BTA} \\
%\text{s.t.} & \quad \frac{1}{N} \tau_{m} \log \left(1+ \frac{\alpha_m^{\{j\}}}{\sigma^{2}} \sum_{k=0}^{N-1} |F_{m,k}|^{2}|G_{m,k}|^{2}P_{m,k}^{\{j\}} \right)  \geq Q, \quad \forall m \\
%&\quad \frac{1}{N} \sum_{m=1}^M \tau_m \sum_{k=0}^{N-1}  \log \left(1+\frac{|H_{k}|^2 P_{m,k}^{\{j\}} }{\alpha_m^{\{j\}}  |F_{m,k}V_{m,k}|^2 P_{m,k}^{\{j\}} + \sigma^2}\right)  \geq D \label{eq:C1P1BTA} \\
%&\quad \eta\sum_{k=0}^{N-1}  |F_{m,k}|^{2} \big[ \tau_{m}  P_{m,k}^{\{j\}} (1-\alpha_m^{\{j\}}) + \sum_{r=1, \neq m}^M  \tau_r P_{r,k}^{\{j\}} \big] \geq E_{\min, m},\quad \forall m  \label{eq:C2P1BTA} \\
%&\quad \sum_{m=1}^M  \sum_{k=0}^{N-1} \tau_{m}P_{m,k}^{\{j\}}\leq \barP \label{eq:C3P1BTA}\\
%&\quad \sum_{m=1}^M \tau_{m}\leq 1 \label{eq:C4P1BTA}\\
%&\quad  \tau_{m}\geq  0, \quad \forall m \label{eq:C5P1BTA}
%\end{align}
%\end{subequations}
\begin{subequations}
\label{eq:P1BTA}
\begin{align}%{\text{(P1-BTA)}}: \
&\underset{Q, \btau}{\max}  \quad Q \label{eq:ObjP1BTA} \\
&\text{s.t.} \quad \eqref{eq:C1P1}, \eqref{eq:C2P1},  \eqref{eq:C3P1},  \eqref{eq:C4P1},  \eqref{eq:C5P1},  \eqref{eq:C7P1},
%\quad \frac{1}{N} \tau_{m} \log \left(1+ \frac{\alpha_m^{\{j\}}}{\sigma^{2}} \sum_{k=0}^{N-1} |F_{m,k}|^{2}|G_{m,k}|^{2}P_{m,k}^{\{j\}} \right)  \geq Q, \quad \forall m \\
%&\quad \frac{1}{N} \sum_{m=1}^M \tau_m \sum_{k=0}^{N-1}  \log \left(1+\frac{|H_{k}|^2 P_{m,k}^{\{j\}} }{\alpha_m^{\{j\}}  |F_{m,k}V_{m,k}|^2 P_{m,k}^{\{j\}} + \sigma^2}\right)  \geq D \label{eq:C1P1BTA} \\
%&\quad \eta\sum_{k=0}^{N-1}  |F_{m,k}|^{2} \big[ \tau_{m}  P_{m,k}^{\{j\}} (1-\alpha_m^{\{j\}}) + \sum_{r=1, \neq m}^M  \tau_r P_{r,k}^{\{j\}} \big] \geq E_{\min, m},\quad \forall m  \label{eq:C2P1BTA} \\
%&\quad \sum_{m=1}^M  \sum_{k=0}^{N-1} \tau_{m}P_{m,k}^{\{j\}}\leq \barP \label{eq:C3P1BTA}\\
%&\quad \sum_{m=1}^M \tau_{m}\leq 1 \label{eq:C4P1BTA}\\
%&\quad  \tau_{m}\geq  0, \quad \forall m \label{eq:C5P1BTA}
\end{align}
\end{subequations}
where the variables $P_{m,k}$'s and $\alpha_m$'s are replaced by given $P_{m,k}^{\{j\}}$'s  and $\alpha_m^{\{j\}}$'s, respectively, in all the constraints. Notice that problem~\eqref{eq:P1BTA} is a standard linear programming (LP), it can be solved efficiently by existing optimization tools such as CVX~\cite{CVXTool2016}.
%Moreover, it can be verified that either the constraint \eqref{eq:C3P1BTA} or \eqref{eq:C4P1BTA} is met with equality when the optimal $\btau$ is obtained for given $\balpha^{\{j\}}$ and $\bP^{\{j\}}$, since otherwise we can always increase $\tau_m$'s without decreasing the objective value.

\vspace{-0.3cm}
\subsection{Reflection Power Allocation Optimization}
%In iteration $j$,
For given backscatter time portions $\btau^{\{j\}}$ and subcarrier power allocation $\bP^{\{j\}}$,  the power reflection coefficients $\balpha$ can be optimized by solving the following problem
\begin{subequations}
\label{eq:P1RPA}
\begin{align}%\text{(P1-RPA)}: \
&\underset{Q, \balpha}{\max}  \quad Q \\
&\text{s.t.} \quad \eqref{eq:C1P1}, \eqref{eq:C2P1},  \eqref{eq:C3P1},  \eqref{eq:C8P1},%\qquad \quad
%\text{s.t.} & \quad \frac{1}{N} \tau_{m}^{\{j\}} \log \left(1+ \frac{\alpha_m}{\sigma^{2}} \sum_{k=0}^{N-1} |F_{m,k}|^{2}|G_{m,k}|^{2}P_{m,k}^{\{j\}} \right)  \geq Q, \quad \forall m \label{eq:C1P1RPA}\\
%&\quad \frac{1}{N} \sum_{m=1}^M \tau_m^{\{j\}} \sum_{k=0}^{N-1}  \log \left(1+\frac{|H_{k}|^2 P_{m,k}^{\{j\}} }{\alpha_m  |F_{m,k}V_{m,k}|^2 P_{m,k}^{\{j\}} + \sigma^2}\right)  \geq D \label{eq:C2P1RPA} \\
%&\quad \eta\sum_{k=0}^{N-1}  |F_{m,k}|^{2} \big[ \tau_{m}  P_{m,k}^{\{j\}} (1-\alpha_m) + \sum_{r=1, \neq m}^M  \tau_r P_{r,k}^{\{j\}} \big] \geq E_{\min, m},\quad \forall m  \label{eq:C3P1RPA} \\
%&\quad  0 \leq \alpha_m \leq  1, \quad \forall m. \label{eq:C4P1RPA}
\end{align}
\end{subequations}
%It can be checked that
where the variables $P_{m,k}$'s and $\tau_m$'s are replaced by given $P_{m,k}^{\{j\}}$'s  and $\tau_m^{\{j\}}$'s, respectively. Since the left-hand-side of the constraint~\eqref{eq:C2P1} with given $P_{m,k}^{\{j\}}$  and $\tau_m^{\{j\}}$ is a decreasing and convex function of $\alpha_m$, the constraint is convex.  Hence, problem \eqref{eq:P1RPA} is a convex optimization problem that can also be efficiently solved by CVX~\cite{CVXTool2016}.
%standard convex optimization solvers such as
%Moreover, given $P_{m,k}^{\{j\}}$'s  and $\tau_m^{\{j\}}$'s, the \eqref{eq:C1P1} is a convex constraint, and \eqref{eq:C3P1} and \eqref{eq:C4P1} are all linear constraints.

\vspace{-0.3cm}
\subsection{Subcarrier Power Allocation Optimization}
%In iteration $j$,
For given backscatter time portions $\btau^{\{j\}}$ and power reflection coefficients $\balpha^{j}$,  the subcarrier power allocation $\bP$ can be optimized by solving the following problem
\begin{subequations}
\label{eq:P1TPA}
\begin{align}%\text{(P1-TPA)}: \
&\underset{Q, \bP}{\max}\quad Q \\
&\text{s.t.}\quad \frac{1}{N} \sum_{m=1}^M \tau_m^{\{j\}} \label{eq:C2P1TPA} \\
&\quad \sum_{k=0}^{N-1}  \log \left(1+\frac{|H_{k}|^2 P_{m,k} }{\alpha_m^{\{j\}}  |F_{m,k}V_{m,k}|^2 P_{m,k} + \sigma^2}\right)  \geq D\nonumber \\
&\quad  \eqref{eq:C1P1}, \eqref{eq:C3P1},  \eqref{eq:C4P1},  \eqref{eq:C6P1},
%\text{s.t.} & \quad \frac{1}{N} \tau_{m}^{\{j\}} \log \left(1+ \frac{\alpha_m^{\{j\}}}{\sigma^{2}} \sum_{k=0}^{N-1} |F_{m,k}|^{2}|G_{m,k}|^{2}P_{m,k} \right)  \geq Q, \quad \forall m \label{eq:C1P1TPA} \\
%&\quad \frac{1}{N} \sum_{m=1}^M \tau_m^{\{j\}} \sum_{k=0}^{N-1}  \log \left(1+\frac{|H_{k}|^2 P_{m,k} }{\alpha_m^{\{j\}}  |F_{m,k}V_{m,k}|^2 P_{m,k} + \sigma^2}\right)  \geq D \label{eq:C2P1TPA}\\
%&\quad \eta\sum_{k=0}^{N-1}  |F_{m,k}|^{2} \big[ \tau_{m}  P_{m,k} (1-\alpha_m) + \sum_{r=1, \neq m}^M  \tau_r^{\{j\}} P_{r,k} \big] \geq E_{\min, m},\quad \forall m  \label{eq:C3P1TPA}\\
%&\quad \sum_{m=1}^M  \sum_{k=0}^{N-1} \tau_{m}^{\{j\}} P_{m,k}\leq \barP \label{eq:C4P1TPA}\\
%&\quad 0 \leq P_{m,k} \leq  P_{\sf peak}, \quad \forall m, \ k \label{eq:C5P1TPA}
\end{align}
\end{subequations}
where the variables $\tau_m$'s and $\alpha_m$'s are replaced by given values $\tau_m^{\{j\}}$'s  and $\alpha_m^{\{j\}}$'s, respectively. Since the constraint function $\tilR(\bP)|_{\btau^{\{j\}}, \balpha^{\{j\}}}$ in \eqref{eq:C2P1TPA} is non-convex with respect to $P_{m, k}$, problem \eqref{eq:P1TPA} is non-convex. Notice that the constraint function $\tilR(\bP)|_{\btau^{\{j\}}, \balpha^{\{j\}}}$ can be rewritten as %a difference of two logarithm functions, as follows , in all the constraints
\begin{align}\label{eq:SCA1}
  &\tilR(\bP)|_{\btau^{\{j\}}, \balpha^{\{j\}}} \nonumber \\
   &= \sum_{m=1}^M \frac{\tau_m^{\{j\}}}{N} \sum_{k=0}^{N-1}  \Big[\!-\!\log \!\left(\! \alpha_m |F_{m,k}V_{m,k}|^2 P_{m,k} + \sigma^2 \!\right)\! +\!...\!  \nonumber \\
  &\quad \log \left( \left(\alpha_m^{\{j\}}  |F_{m,k}V_{m,k}|^2 + |H_{k}|^2 \right) P_{m,k} + \sigma^2 \right) \Big].
\end{align}

%successive convex optimization (
To handle the non-convex constraint \eqref{eq:C2P1TPA}, we exploit the SCO technique \cite{Beck2010} to approximate the second logarithm function in \eqref{eq:SCA1}. Recall that any concave function can be globally upper-bounded by its first-order Taylor expansion at any point. Specifically, let $P_{m,k}^{\{j\}}$ denote the subcarrier power allocation in the last iteration. We have the following concave lower bound at the local point $P_{m,k}^{\{j\}}$
  %for the non-convex constraint function in \eqref{eq:C2P1TPA}
\begin{align}
%&= \frac{1}{N} \sum_{m=1}^M \tau_m^{\{j\}} \sum_{k=0}^{N-1}  \left[ \log \left( \left(\alpha  |F_{m,k}V_{m,k}|^2 + |H_{k}|^2 \right) P_{m,k} + \sigma^2 \right)  - \log \left(\alpha  |F_{m,k}V_{m,k}|^2 P_{m,k} + \sigma^2\right)  \right] \nonumber \\
  &\tilR(\bP)|_{\btau^{\{j\}}, \balpha^{\{j\}}, \bP^{\{j\}}}  \label{eq:SCALB} \!\geq\!\\
  \!& \sum_{m=1}^M\! \frac{\tau_m^{\{j\}}}{N} \!\sum_{k=0}^{N-1}\!  \Big[ \!\log\! \left( \!\left(\!\alpha  |F_{m,k}V_{m,k}|^2 \!+\! |H_{k}|^2 \!\right)\! P_{m,k} \!+\! \sigma^2 \right) \!-\! ... \nonumber \\
  &\log \left(\alpha  |F_{m,k}V_{m,k}|^2 P_{m,k}^{\{i, j\}} + \sigma^2\right)  - \nonumber \\
  &\frac{\alpha  |F_{m,k}V_{m,k}|^2 (P_{m,k}\!-\! P_{m,k}^{\{i, j\}})}{\alpha  |F_{m,k}V_{m,k}|^2 P_{m,k}^{\{i, j\}} \!+\! \sigma^2}\Big] \!\triangleq\! \tilR^{\sf lb}(\bP)|_{\btau^{\{j\}}, \balpha^{\{j\}}, \bP^{\{j\}}}.\nonumber
\end{align}

With given local points $\bP^{\{j\}}$ and lower bound $\tilR^{\sf lb}(\bP)|_{\btau^{\{j\}}, \balpha^{\{j\}}, \bP^{\{j\}}}$ in \eqref{eq:SCALB}, by introducing the lower-bound minimum-throughput $Q_{\sf tpa}^{\sf lb}$, problem \eqref{eq:P1TPA} is approximated as the following problem %(P1-TPA)
\begin{subequations}
\label{eq:P1TPAA}
\begin{align}%\text{(P1-TPA-Appr.)}: \
&\underset{Q_{\sf tpa}^{\sf lb}, \bP}{\max}  \quad Q_{\sf tpa}^{\sf lb} \\
&\text{s.t.} \quad \tilR^{\sf lb}(\bP)|_{\btau^{\{j\}}, \balpha^{\{j\}}, \bP^{\{j\}}} \geq D \label{eq:C1P1TPAA} \\
&\qquad  \eqref{eq:C1P1}, \eqref{eq:C3P1},  \eqref{eq:C4P1},  \eqref{eq:C6P1},
%\text{s.t.} & \quad \frac{1}{N} \tau_{m}^{\{j\}} \log \left(1+ \frac{\alpha_m^{\{j\}}}{\sigma^{2}} \sum_{k=0}^{N-1} |F_{m,k}|^{2}|G_{m,k}|^{2}P_{m,k} \right)  \geq Q_{\sf tpa}^{\sf lb}, \quad \forall m \label{eq:C1P1TPAA}\\
%&\quad \frac{1}{N} \sum_{m=1}^M \tau_m^{\{j\}} \sum_{k=0}^{N-1}  \Big[ \log \left( \left(\alpha_M^{\{j\}}  |F_{m,k}V_{m,k}|^2 + |H_{k}|^2 \right) P_{m,k} + \sigma^2 \right)-  ... \label{eq:C2P1TPAA} \\
%  &\quad \log \left(\alpha_m^{\{j\}}  |F_{m,k}V_{m,k}|^2 P_{m,k}^{\{j\}} + \sigma^2\right)  - \frac{\alpha_m^{\{j\}}  |F_{m,k}V_{m,k}|^2 (P_{m,k}- P_{m,k}^{\{j\}})}{\alpha_m^{\{j\}} |F_{m,k}V_{m,k}|^2 P_{m,k}^{\{j\}} + \sigma^2} \Big]  \geq D, \nonumber \\
%&\quad \eta\sum_{k=0}^{N-1}  |F_{m,k}|^{2} \big[ \tau_{m}^{\{j\}}  P_{m,k} (1-\alpha_m^{\{j\}}) + \sum_{r=1, \neq m}^M  \tau_r^{\{j\}} P_{r,k} \big] \geq E_{\min, m},\quad \forall m  \label{eq:C3P1TPAA} \\
%&\quad \sum_{m=1}^M  \sum_{k=0}^{N-1} \tau_{m}^{\{j\}} P_{m,k}\leq \barP \label{eq:C4P1TPAA}\\
%&\quad 0 \leq P_{m,k} \leq  P_{\sf peak}, \quad \forall m, \ k \label{eq:C5P1TPAA}
\end{align}
\end{subequations}
where the variables $\tau_m$'s and $\alpha_m$'s are replaced by given $\tau_m^{\{j\}}$'s  and $\alpha_m^{\{j\}}$'s, respectively. Problem \eqref{eq:P1TPAA} is a convex optimization problem which can also be efficiently solved by CVX~\cite{CVXTool2016}. It is noticed that the lower bound adopted in \eqref{eq:C1P1TPAA} implies that the feasible set of problem \eqref{eq:P1TPAA} is always a subset of that of problem \eqref{eq:P1TPA}. As a result, the optimal objective value obtained from problem \eqref{eq:P1TPAA} is in general a lower bound of that of problem \eqref{eq:P1TPA}.
% (P1-TPA-Appr.)

\subsection{Overall Algorithm}
%From the results in the previous three subsections, , also known as the alternating optimization method
We propose an overall iterative algorithm for problem \eqref{eq:P1} by applying the BCD technique~\cite{HongLuoBCDSPM17}. Specifically, the entire variables in original problem \eqref{eq:P1} are partitioned into three blocks, i.e., $\btau, \balpha$, and $\bP$, which are alternately optimized by solving problem \eqref{eq:P1BTA}, \eqref{eq:P1RPA}, and \eqref{eq:P1TPAA} correspondingly in each iteration, while keeping the other two blocks of variables fixed. Furthermore, the obtained solution in each iteration is used as the input of the next iteration. The details are summarized in Algorithm \ref{AlgorithmP1}.
%. Then, the backscatter time portions $\btau$, power reflection coefficients $\balpha$, and subcarrier power $\bP$

\begin{algorithm}[t!]
\caption{Iterative Algorithm for solving problem \eqref{eq:P1}}\label{AlgorithmP1}
\begin{algorithmic}[1]
\STATE Initialize $Q^{\{0\}}, \ Q^{\{1\}}, \ \epsilon, \ \btau ^{\{0\}}, \ \balpha ^{\{0\}}, \bP^{\{0\}}$. Let $j=0$. \\
%$j=0, Q^{\{0\}}=10^{-3}, Q^{\{1\}}=1, \epsilon =10^{-4}, \tau ^{\{0\}}_m=\frac{1}{M}, \alpha ^{\{0\}}_m=0.5, P^{\{0\}}_{k,m}=\frac{1}{MN}, \forall k, \ m$.
\WHILE {($|Q^{\{j+1\}}-Q^{\{j\}}|>\epsilon$)}{
\STATE Solve problem \eqref{eq:P1BTA} for given $\balpha^{\{j\}}$ and $\bP^{\{j\}}$, and obtain the optimal solution as $\btau^{\{j+1\}}$.
\STATE Solve problem \eqref{eq:P1RPA} for given $\btau^{\{j+1\}}$ and $\bP^{\{j\}}$, and obtain the optimal solution as $\balpha^{\{j+1\}}$.
\STATE Solve problem \eqref{eq:P1TPAA} for given $\btau^{\{j+1\}}$, $\balpha^{\{j+1\}}$, and $\bP^{\{j\}}$, and obtain the optimal solution as $\bP^{\{j+1\}}$.
}
\STATE  Update iteration index $j=j+1$.
\ENDWHILE
\STATE  Return optimal solution $\btau^{\star \{j\}}$, $\balpha^{\star \{j\}}$, and $\bP^{\star \{j\}}$, and objective value $Q^{\star \{j\}}(\btau^{\star \{j\}}, \balpha^{\star \{j\}},\bP^{\star \{j\}})$.
\end{algorithmic}
\end{algorithm}
%\vspace{-0.2cm}

\subsection{Convergence and Complexity Analysis}
%It is worth pointing out that in the classical BCD method, the sub-problem for updating each block of variables is required to be solved exactly with optimality in each iteration in order to guarantee the convergence~\cite{HongLuoBCDSPM17}. However,
Notice that in our case, for subcarrier power allocation problem \eqref{eq:P1TPA}, we only solve its approximate problem \eqref{eq:P1TPAA} optimally. Thus, the convergence analysis for the classical BCD technique cannot be directly applied~\cite{HongLuoBCDSPM17}, and the convergence of Algorithm \ref{AlgorithmP1} needs to be proved, as follows.
\begin{mythe}
Algorithm \ref{AlgorithmP1} is guaranteed to converge.
\end{mythe}
% of problem \eqref{eq:P1}

\begin{proof}
First, in step 3 of Algorithm \ref{AlgorithmP1}, since the optimal solution $\btau^{\{j+1\}}$ is obtained for given $\balpha^{\{j\}}$ and $\bP^{\{j\}}$, we have the following inequality on the minimum throughput $Q(\btau, \balpha, \bP)$
\begin{align}
  Q(\btau^{\{j\}}, \balpha^{\{j\}}, \bP^{\{j\}}) \leq Q(\btau^{\{j+1\}}, \balpha^{\{j\}}, \bP^{\{j\}}). \label{eq:Qinequality1}
\end{align}

% of Algorithm \ref{AlgorithmP1}
Second, in step 4, since the optimal solution $\balpha^{\{j+1\}}$ is obtained for given $\btau^{\{j+1\}}$ and $\bP^{\{j\}}$, it holds that
\begin{align}
  Q\!(\!\btau^{\{j+1\}}, \balpha^{\{j\}}, \bP^{\{j\}} \!)\! \leq Q\!(\!\btau^{\{j+1\}}, \balpha^{\{j+1\}}, \bP^{\{j\}}\!)\!.\label{eq:Qinequality2}
\end{align}

%For convenience, define $Q^{\sf lb}_{\sf tpa} (\btau^{\{j+1\}}, \balpha^{\{j+1\}}, \bP^{\{j\}})$ as the objective function of problem  \eqref{eq:P1TPAA} based on $(\btau^{\{j+1\}}, \balpha^{\{j+1\}}$, and $\bP^{\{j\}})$.

Third, in step 5, it follows that
%since the optimal solution $\balpha^{\{j+1\}}$ is obtained for given $\btau^{\{j\}}$ and $\bP^{\{j\}}$, we have
\begin{align}
  Q\!(\!\btau^{\{j+1\}}\!, \! \balpha^{\{j+1\}}\!, \! \bP^{\{j\}}\!)\! &\!\eqa\! Q^{\sf {lb}, \{j\}}_{\sf tpa} \!(\!\btau^{\{j+1\}}, \balpha^{\{j+1\}}, \bP^{\{j\}}\!)\! \nonumber \\
  &\!\leb\!  Q^{\sf {lb}, \{j\}}_{\sf tpa} \!(\!\btau^{\{j+1\}}, \balpha^{\{j+1\}}, \bP^{\{j+1\}}\!)\! \nonumber \\
  &\!\lec\! Q\!(\!\btau^{\{j+1\}}\!, \balpha^{\{j+1\}}\!, \bP^{\{j+1\}}\!)\!, \label{eq:Qinequality3}
\end{align}
where (a) holds since the Taylor expansion in \eqref{eq:SCALB} is tight at given local point, which implies problem \eqref{eq:P1TPAA} at $\bP^{\{j\}}$ has the same objective function as that of problem \eqref{eq:P1TPA}; (b) is because $\bP^{\{j+1\}}$ is the optimal solution to problem \eqref{eq:P1TPAA}; (c) holds since the objective value of problem \eqref{eq:P1TPAA} is a lower bound of that of its original problem \eqref{eq:P1TPA}.
%The inequality in \eqref{eq:Qinequality3} indicates that the objective value is always non-decreasing after each iteration, although an approximated optimization problem \eqref{eq:P1TPAA} is solved to obtain the optimal subcarrier power $\bP$ in each iteration.

From \eqref{eq:Qinequality1}, \eqref{eq:Qinequality2}, and \eqref{eq:Qinequality3}, we have
\begin{align}
   Q(\!\btau^{\{j\}}, \balpha^{\{j\}}, \bP^{\{j\}})\! \leq Q(\!\btau^{\{j+1\}}, \balpha^{\{j+1\}}, \bP^{\{j+1\}})\!,
\end{align}
which implies that the objective value of problem \eqref{eq:P1} is non-decreasing after each iteration in Algorithm \ref{AlgorithmP1}. Since the objective value of problem \eqref{eq:P1} is a finite positive value, the proposed Algorithm \ref{AlgorithmP1} is guaranteed to converge to the optimal objective value and solutions. This completes the convergence proof.
\end{proof}

%Finally, it is noted that t
The complexity of Algorithm \ref{AlgorithmP1} is polynomial, since only three convex optimization problems need to be solved in each iteration. Hence, the proposed Algorithm \ref{AlgorithmP1} can be practically implemented with fast convergence for full-duplex ABCNs with a moderate number of BDs and LU(s).

\vspace{-0.2cm}
\section{NUMERICAL RESULTS}\label{simulation}
In this section, we provide simulation results to evaluate the performance of the proposed joint design. We consider an ABCN with $M=2$ BDs. Suppose that the FAP-to-BD1 distance and FAP-to-BD2 distance are 2.5 m and 4 m, respectively, the FAP (BD1, BD2)-to-LU distances are all 15 m. We assume independent Rayleigh fading channels, i.e., the channel coefficient of each path is a circularly symmetric complex Gaussian prandom variable, and the power gains of multiple paths are exponentially distributed. For each channel link, its first-path channel power gain is assumed to be $10^{-3} d^{-2}$, where $d$ is the distance with unit of meter. Let the number of pathes $L_{\sf f}=L_{\sf g}=4$, $L_{\sf h}=8$, and $L_{\sf v}=6$. Other parameters are set as $N=64, N_{\sf cp}=16, \barP=1, \eta=0.5, \epsilon=10^{-4}$. Define the average receive SNR at the FAP as $\bar{\gamma}=\barP \sum \nolimits_{l=0}^{L_{\sf f}-1} \bbE[\left| {{g_{1,l}f_{1, l}}} \right|^2]/\sigma^2$. Let $E_{\min, 1}=E_{\min, 2}=E_{\min}$. For performance comparison, we consider a benchmark scheme in which the backscatter time portion and subcarrier power are equally allocated, i.e., $\tau_m=\frac{1}{M}, P_{m,k}=P_{\sf ave}=\frac{1}{MN}$, and all BDs adopt a common power reflection coefficient that is optimized via CVX. The following results are obtainepd based on 100 random channel realizations.
% between the transmitting and the receiving antennas  via one-dimension search

Fig.~\ref{fig:FigSim1} plots the max-min throughput of all BDs versus the LU's throughput requirement $D$ for different SNRs $\bar{\gamma}$'s. We fix $P_{\sf peak} = 20 P_{\sf ave}$ and $E_{\min}=10 \  \mu$J. As expected, the max-min throughput decreases as $D$ increases, which reveals the throughput tradeoff between the BDs and the LU. We further observe that the max-min throughput performance is significantly enhanced by using the proposed joint design, compared to the benchmark scheme. Also, higher max-min throughput is achieved when the SNR at the FAP is higher.
%hwe observe that for higher LU throughput requirement $D$, the max-min throughput of BDs is decreases. This throughput gap due to higher $D$ decreases as the SNR increases.
\begin{figure} [!t]
\vspace{-0.2cm}
	\centering	\includegraphics[width=.98\columnwidth]{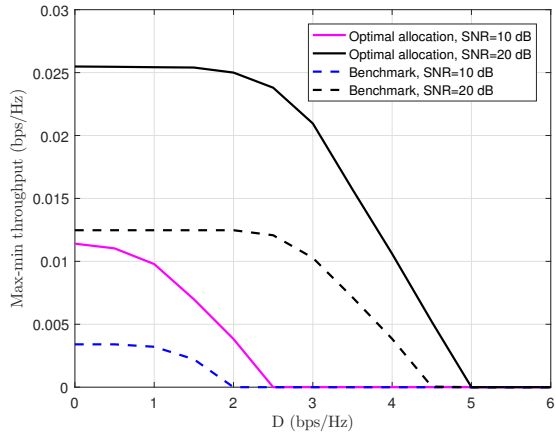}
	\caption{Max-min throughput versus LU throughput requirement for different SNRs.} \label{fig:FigSim1}
\vspace{-0.6cm}
\end{figure}

%\begin{figure} [h!]
%\vspace{-0.2cm}
%	\centering	\includegraphics[width=.99\columnwidth]{MaxMin_Emin_Region_V180429A_SNR20_D012_7Chs_V2.eps}
%\caption{Max-min throughput versus BD energy requirement.} \label{fig:FigSim2}
%\vspace{-0.5cm}
%\end{figure}

Fig.~\ref{fig:FigSim2} compares the max-min throughput under different BDs' energy requirements $E_{\min}$'s and subcarrier peak-power values $P_{\sf peak}$'s, for both the proposed joint design and the benchmark scheme. We fix $D=1$ bps/Hz. In general, the max-min throughput increases as the SNR $\bar{\gamma}$ increases. We have three further observations. First and foremost, the proposed joint design achieves significant max-min throughput gains as compared to the benchmark scheme. Second, higher max-min throughput is achieved for lower harvested-energy requirement $E_{\min}$ with given $P_{\sf peak}$, which reveals the BDs' throughput-energy tradeoff. This observation can be specifically obtained from the three red solid curves for our proposed joint design and the three blue dotted curves for the benchmark scheme, given $P_{\sf peak} = 20 P_{\sf ave}$. Third, higher max-min throughput is obtained for higher peak-power value $P_{\sf peak}$ with given $E_{\min}$, which is demonstrated in the red and black solid curves with triangle marker for our proposed joint design.

\begin{figure} [h!]
\vspace{-0.2cm}
	\centering	\includegraphics[width=.98\columnwidth]{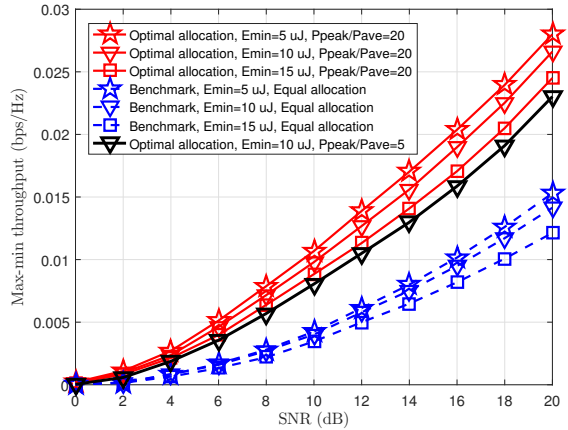}
%	\centering	\includegraphics[width=.99\columnwidth]{MaxMin_V20180429_Ppeak_D1_Emin10muJ_10Chs.eps}
%MaxMin_V20180429_Ppeak_D1_Emin10muJ_10Chs_V2font
		\caption{Max-min throughput versus SNR for different harvested-energy requirements and peak power values.} \label{fig:FigSim2}
\vspace{-0.5cm}
\end{figure}

\vspace{-0.2cm}
\section{Conclusion}\label{conslusion}
This paper has investigated a full-duplex AmBC network over ambient OFDM carriers. The minimum throughput among all BDs is maximized by jointly optimizing the BDs' backscatter time portions, the BDs' power reflection coefficients, and the FAP's subcarrier power allocation. By utilizing the block coordinated decent and successive convex optimization techniques, an efficient iterative algorithm is proposed, which is guaranteed to converge. Numerical results show that significant throughput gains are achieved as compared to the benchmark scheme with equal resource allocation, benefitting from multiple design dimensions of the proposed joint optimization. The interesting BDs' throughput-energy tradeoff and the throughput tradeoff between the BDs and the LU are also revealed.

%\vspace{-0.2cm}
 \renewcommand{\baselinestretch}{0.87}
\bibliography{IEEEabrv,reference1804}%1604
\bibliographystyle{IEEEtran}%This sentence must appear after \bibliography statement: commented on April-28-2018

\end{document}